\documentclass[11pt]{article}

\usepackage{amsmath}
\usepackage{amsfonts}
\usepackage{amssymb}
\usepackage{amscd}
\usepackage{amsthm}
\usepackage{latexsym}

\def\H{\mathcal H}
\def\K{\mathcal K}
\def\P{\mathbb P}
\def\v{\mathfrak v}
\def\Tr{{\rm Tr}}
\def\d{\mathrm{d}}
\def\F{\mathcal F}
\def\S{\mathcal S}
\def\y{\textsc y}
\def\A{{\mathfrak A}}
\def\CC{\mathfrak C}
\def\BBh{{\mathfrak B}^{B}_\hb}
\def\BB{\mathcal B}
\def\X{\mathcal X}
\def\C{\mathbb{C}}
\def\N{\mathbb{N}}
\def\R{\mathbb{R}}

\def\B{{\mathcal B}}
\def\M{\mathcal M}

\def\e{\mathfrak e}
\def\im{\mathop{\mathsf{Im}}\nolimits}

\def\Op{\mathfrak{Op}^A_\hb}
\def\ABh{{\mathfrak A}_\hb^B}
\def\hb{\hbar}

\newtheorem{lemma}{Lemma}[section]

\newtheorem{theorem}[lemma]{Theorem}
\newtheorem{proposition}[lemma]{Proposition}
\newtheorem{definition}[lemma]{Definition}
\newtheorem{remark}[lemma]{Remark}

\numberwithin{equation}{section}

\begin{document}

\title{Coherent states in the presence of a variable magnetic field}

\date{\today}

\author{Marius M\u antoiu$^1$ $^2$, Radu Purice$^2$ and Serge Richard$^3\footnote{On leave from Universit\'e de Lyon;
Universit\'e Lyon 1; CNRS, UMR5208, Institut Camille Jordan, 43 blvd du 11 novembre 1918, F-69622 Villeurbanne-Cedex, France.} $}
\date{\small}
\maketitle
\vspace{-1cm}

\begin{quote}
\emph{
\begin{itemize}
\item[$^1$] Departamento de Matem\'aticas, Universidad de Chile, Las Palmeras 3425, Casilla 653,
Santiago, Chile
\item[$^2$]Institute
of Mathematics "Simion Stoilow" of the Romanian Academy, P.O.  Box
1-764, Bucharest, RO-70700, Romania
\item[$^3$] Department of Pure Mathematics and Mathematical Statistics,
Centre for Mathematical Sciences, University of Cambridge,
Cambridge, CB3 0WB, United Kingdom
\item[] \emph{E-mails:}  mantoiu@imar.ro, purice@imar.ro, sr510@cam.ac.uk
\end{itemize}
  }
\end{quote}

\textbf{2000 Mathematics Subject Classification:} 81R30, 46L89.

\textbf{Key Words:}  Coherent state, quantization, magnetic field,
 magnetic pseudodifferential operators.

\maketitle \abstract{We introduce magnetic coherent states for a
particle in a variable magnetic field. They provide a pure state quantization of the phase space
$\R^{2N}$ endowed with a magnetic symplectic form.}

\section*{Introduction}

This article is concerned with a spinless non-relativistic particle placed in a variable magnetic field.
Recent publications
\cite{IMP,KO1,KO2,LMR,MP1,MP2,MPR1,Mu} introduced and developed a
mathematical formalism for the observables naturally associated with such a system,
both in a classical and in a quantum framework.
A brief survey of this topic has been exposed in \cite{MP3}.
We now would like to complete the picture, indicating an appropriate way to model the states
of these systems.

Classically, the magnetic field changes the geometry of the phase-space and this is implemented by
a modification of the standard symplectic form.
Consequently, it also modifies the Poisson algebra structure of the smooth functions on phase-space,
which are interpreted as classical observables, see \cite{MR,MP2,MP3} for details.
At the quantum level, one introduces algebras of observables defined only in terms of the magnetic field
\cite{KO1,KO2,MP1,MP3,Mu}. The main new feature is a composition law on symbols defined in terms of fluxes
of the magnetic field through triangles.
With a proper implementation of Planck's constant $\hb$, it has been proved in \cite{MP2} that
the quantum algebra of observables converges to the classical one in the
sense of strict deformation quantization \cite{La3,Ri1,Ri3,Ri4}.

To get the traditional setting involving self-adjoint operators,
the quantum algebra can be represented in Hilbert spaces.
This is realised by choosing any vector potential defining the magnetic field.
In such a way one gets essentially a new pseudodifferential calculus \cite{IMP,IMP',LMR,MP1}
which can be interpreted as a functional calculus for the family of non-commuting operators composed
of positions and magnetic momenta.
When no magnetic field is present, it coincides with the Weyl quantization.
In has been adapted in \cite{BB} to the framework of nilpotent Lie groups.
One of main virtue of this construction is gauge-covariance: equivalent choices of vector potentials
lead to unitarily equivalent representations.
Both the intrinsic and the represented version admit $C^*$-algebraic reformulations \cite{MP3,MPR1}.
They were essential in \cite{MP2} to prove the classical norm-sense limit of the quantum structure
and are also useful in the spectral analysis of magnetic Schr\"odinger operators \cite{MPR2}.

Now, the complete formalism involves coupling the classical and the quantum observables to states.
We refer to \cite{La1,La2,La3,La4} and to references therein for a general presentation and justification
of the concept of state quantization.
We simply recall that the space of pure states of both classical and quantum mechanical systems
are Poisson spaces with a transition probability \cite[Def.~I.3.1.4]{La3}.
In the magnetic case, the classical setting consists in the phase space $\Xi:=\R^{2N}$
endowed with the magnetic symplectic form $\sigma^B$  and the transition probability
defined by
\begin{equation*}
p^{\rm cl}:\Xi\times\Xi\rightarrow [0,1],\qquad p^{\rm cl}(X,Y):=\delta_{XY}\ .
\end{equation*}
On another hand, the pure states space of $\K(\H)$ (the $C^*$-algebra of all the compact operators
in the Hilbert space $\H$) with the $w^*$-topology is homeomorphic to the projective space $\P(\H)$
with its natural topology, see \cite[Prop.~I.2.5.2]{La3} and \cite[Corol.~I.2.5.3]{La3}.
The latter space is also endowed with the $\hb$-dependent Fubini-Study symplectic form $\Sigma'_{\hb}$.
With the interpretation of elements $\v$ of $\P(\H)$ as one dimensional orthogonal projections
$|v\rangle \langle v|$ defined by unit vectors $v \in \H$, the quantum transition probability is given by
\begin{equation}\label{tranzactia}
p^{\rm qu}:\P(\H)\times\P(\H)\rightarrow [0,1],\qquad
p^{\rm qu}(\mathfrak u,\v):=\Tr(|u\rangle \langle u|\!\cdot\!|v\rangle \langle v|)=|\langle u,v\rangle|^2\ .
\end{equation}

A pure state quantization corresponds then to a family of injective embeddings
$\{\v_{\hb}:\Xi\rightarrow \P(\H)\}_{\hb\in (0,1]}$
satisfying a certain set of axioms \cite[Def.~II.1.3.3]{La3}.
In particular, the transition probabilities and the symplectic structures of $\Xi$ and $\P(\H)$
are respectively connected in the limit $\hb\rightarrow 0$.
In our case the embeddings $\v_\hb\equiv \v_\hb^A$ are defined in Definition \ref{enfin}
by the choice of a vector potential
$A$ generating the magnetic field. When composed with the magnetic Weyl calculus,
they furnish pure states $\v_\hb^B$ on a $C^*$-algebra defined intrinsically by the magnetic composition law.
At this level the pure states only depend on the magnetic field and not on any vector potential.
In Theorem \ref{themainTh} we sum up how such a pure state quantization can be achieved in the magnetic case.

In conformity with \cite[Def.~II.1.5.1]{La3} this pure state quantization is coherent,
{\it i.e.}~it can be deduced from a family of continuous injections $v^A_{\hb}:\Xi\rightarrow \H$,
see Definition \ref{bunbun}. In a future article we are going to pursue this,
introducing and studying the Berezin-Toeplitz type quantization associated with these coherent states
as well as the corresponding Bargmann transform.

The structure of this article is the following: In the first section we make a short survey
of the quantization of observables. In Section \ref{oita} we define coherent vectors and coherent states and
study their properties, especially those connected to the limit $\hb\rightarrow 0$.
In the last section we put the results of Section \ref{oita} in the perspective of deformation quantization
of states and observables.

\begin{center}
 --------------------------------------------------------------------------------------------
\end{center}

 \medskip
 {\bf Acknowledgements:} M. M\u antoiu is partially supported by {\it N\'ucleo Cient\' ifico ICM P07-027-F
 "Mathematical Theory of Quantum and Classical Magnetic Systems"} and by the Chilean Science Fundation {\it Fondecyt}
 under the grant 1085162.

 R. Purice acknowledges support from the ANCS Contract No. 2 CEx06-11-18/2006.

 S. Richard is supported by the Swiss National Science Foundation.

\section{Quantization of observables}\label{vacoita}

In this Section we recall the structure of the observable algebras of a particle in a variable magnetic field,
both from a classical and a quantum point of view.
We follow the references \cite{MP1,MP2,MP3} which contain further details and technical developments.
Our main purpose is to introduce the basic objects that will be used subsequently and to give motivations.

Let us consider the physical system consisting in a spinless particle moving in the Euclidean space
$\X:=\R^N$ under the influence of a magnetic field. We denote by $\X^*$ the dual space of $\X$.
The duality is simply given by $\X\times\X^*\ni(x,\xi)\mapsto x\cdot\xi$.
The phase space is denoted by $\Xi:=T^*\X\equiv\X\times\X^*$; systematic notations as
$X=(x,\xi)$, $Y=(y,\eta)$, $Z=(z,\zeta)$ will be used for its points.
If no magnetic field is present, the standard symplectic form
\begin{equation}\label{simp}
\sigma(X,Y)\equiv\sigma[(x,\xi),(y,\eta)]:=y\cdot \xi-x\cdot\eta
\end{equation}
prepares $\Xi$ for doing classical mechanics.

The magnetic field is a closed $2$-form $B$ on $\X$ ($\d B=0$), given by the matrix-components
$$
B_{jk}=-B_{kj}:\X\rightarrow \R\qquad j,k=1,\dots,N.
$$
Suitable smoothness will be indicated when necessary; for the moment we simply assume that the components of the
magnetic field are continuous.
The effect of $B$ is to change the geometry of the phase space by adding an extra term to \eqref{simp}:
$\sigma^B:=\sigma+\pi^*B,$ where $\pi^*$ is the pull-back
associated with the cotangent bundle projection $\pi:T^*\X\rightarrow\X$. In coordinates
\begin{equation}\label{simpB}
(\sigma^B)_X(Y,Z)=z\cdot \eta -y\cdot\zeta+B(x)(y,z)=\sum_{j}(z_j\;\!\eta_j-y_j\;\!\zeta_j)+
\sum_{j,k}B_{jk}(x)\;\!y_j\;\!z_k.
\end{equation}

If the components of the magnetic field are smooth, one can associate with this new symplectic form
a canonical Poisson bracket acting on functions $f,g\in C^\infty(\Xi)$ by the formula:
\begin{equation}\label{pison}
\{f,g\}^B=\sum_{j}(\partial_{\xi_j}f\,\partial_{x_j}g-\partial_{\xi_j}g\,\partial_{x_j}f)+
\sum_{j,k}B_{jk}(\cdot)\,\partial_{\xi_j}f\,\partial_{\xi_k}g.
\end{equation}
It is a standard fact that $C^\infty(\Xi;\R)$ endowed with
$\{\cdot,\cdot\}^B$ and with the usual product of functions is {\it a
Poisson algebra}. This means that $\{\cdot,\cdot\}^B$ is a Lie
bracket, $(f,g)\mapsto f g$ is bilinear, associative and
commutative and the Leibnitz rule $\{f,g h\}^B=\{f,g\}^B h+g \{f,h\}^B$
holds for all $f,g,h\in C^{\infty}(\Xi;\R)$.

The point of view of deformation quantization is to suitably modify the classical structure of a Poisson algebra to get a
quantum structure of observables for values of some Planck constant $\hb\in I:=(0,1]$. Eventually, for $\hb\mapsto 0$, the
Poisson algebra will re-emerge in some sense.

The magnetic field $B$ comes into play in defining the observables composition in terms of its fluxes through triangles.
If $a,b,c\in \X$, then we denote by $\langle a,b,c \rangle$ the triangle in $\X$ of vertices $a,b$ and $c$ and set
$$
\Gamma^B\langle a,b,c \rangle:=\int_{\langle a,b,c\rangle}B
$$
for the flux of $B$ through this triangle (invariant integration of a $2$-form through a $2$-simplex).
With this notation and for $f,g: \Xi \to \C$, the formula
\begin{equation}\label{composition}
\left(f\sharp^B_{\hb} g\right)(X)\,:=\,(\pi\hb)^{-2N}\int_\Xi
\int_\Xi \d Y \;\!\d Z\;\!e^{-\frac{2i}{\hb} \sigma(X-Y,X-Z)}\;\!
 e^{-\frac{i}{\hb}\Gamma^B\langle x-y+z,y-z+x,z-x+y\rangle}\;\!f(Y)\;\!g(Z)
\end{equation}
defines a formal associative composition law on functions.
For $B=0$ it coincides with the Weyl composition of symbols in pseudodifferential theory.

The formula \eqref{composition} makes sense and has nice properties under various circumstances. For example, if the
components $B_{jk}$ belong to $C^\infty_{\rm{pol}}(\X)$, the class of smooth functions on $\X$ with polynomial
bounds for all the derivatives, then the Schwartz space $\S(\Xi)$ is stable under $\sharp^B_\hb$.

By denoting by $\M^B_\hb(\Xi)$ the largest space of tempered distributions for which
$$
\sharp^B_{\hb}:\S(\Xi)\times\M^B_\hb(\Xi)\rightarrow\S(\Xi)\quad \rm{and}\quad
\sharp^B_{\hb}:\M^B_\hb(\Xi)\times\S(\Xi)\rightarrow\S(\Xi),
$$
it has be shown in \cite{MP1} that $\M^B_\hb(\Xi)$ is a $^*$-algebra under $\sharp^B_{\hb}$ and under
complex conjugation.
This is a large class of distributions, containing all the bounded measures as well as the class
$C^\infty_{\rm{pol,u}}(\Xi)$
of all smooth functions for which all the derivatives are bounded by some polynomial
(depending on the function, but not on the order of the derivative).
In addition, if we assume that all the derivatives of the functions $B_{jk}$ are bounded,
then the H\"ormander classes of symbols $S^m_{\rho,\delta}(\Xi)$ compose in the usual way under $\sharp^B_{\hb}$.

The most satisfactory approach to introduce norms would be by using twisted $C^*$-dynamical systems and
twisted crossed products, as in \cite{MP2,MP3,MPR1}. To spare space and especially to avoid
using non-trivial facts about $C^*$-algebras, we shall borrow the needed structures from representations.

Being a closed $2$-form in $\X=\R^N$, the magnetic field is exact: it can be written as $B=\d A$ for some $1$-form
$A$ (called {\it vector potential}). It is easy to see that if $B$ is of class $C^\infty_{\rm{pol}}(\X)$,
then $A$ can be chosen in the same class and we shall assume this in the sequel.
The vector potentials enter into the construction by their circulations
$\Gamma^A[x,y]:=\int_{[x,y]}A$ along segments $[x,y]:=\{tx+(1-t)y\mid t\in[0,1]\}$ for any $x,y\in\X$.
For a vector potential $A$ with $\d A=B$ and for $u : \X \to \C$, let us define
\begin{equation}\label{op}
\big[\Op(f)u\big](x):=(2\pi\hb)^{-N}\int_\X \int_{\X^*}\d y\;\! \d\eta \;\!
e^{\frac{i}{\hb}(x-y)\cdot\eta}\;\!
e^{-\frac{i}{\hb}\,\Gamma^A[x,y]}\;\!{\textstyle f\left(\frac{x+y}{2},\eta\right)}\;\!u(y).
\end{equation}
For $A=0$ one recognizes the Weyl quantization, associating with functions or distributions on $\Xi$ linear
operators acting on function spaces on $\X$.

It has been shown that $\Op$, suitably interpreted (by using rather simple duality arguments)
defines a representation of the $^*$-algebra $\M^B_\hb(\Xi)$ by linear continuous operators
$:\S(\X)\rightarrow\S(\X)$. This means, of course, that
$$
\Op\big(f\sharp ^B_\hb g\big)=\Op(f)\;\!\Op(g)\quad {\rm and}\quad \Op(\overline f)=\Op(f)^*
$$
for any $f,g\in\M^B_\hb(\Xi)$. In addition, $\Op$ restricts to an isomorphism from $\S(\Xi)$
to $\BB[\S^*(\X),\S(\X)]$ and extends to an isomorphism from $\S^*(\Xi)$ to $\BB[\S(\X),\S^*(\X)]$
(we set $\BB(\mathcal R,\mathcal T)$ for the family of all linear continuous operators between the topological
vector spaces $\mathcal R$ and $\mathcal T$).
The class of Hilbert-Schmidt operators in the Hilbert space $\H:=L^2(\X)$
coincides with the class of operators of the form $\Op(f)$ with $f \in L^2(\Xi)$.

An important property of \eqref{op} is {\it gauge covariance}. If $A'=A+\d \rho$ for a smooth function $\rho: \X \to \R$ (the equivalent choices $A$ and $A'$  would give the same magnetic field),
then $\mathfrak{Op}^{A'}_\hb(f)=e^{\frac{i}{\hb}\rho}\;\!\Op(f)\;\!e^{-\frac{i}{\hb}\rho}$.
Such a unitary equivalence would not hold for the wrong quantization appearing in the literature
$$
\big[\mathcal{O}p_{A,\hb}(f)u\big](x):=(2\pi\hb)^{-N}\int_\X \int_{\X^*}\d y \;\! \d \eta\;\!
e^{\frac{i}{\hb}(x-y)\cdot\eta}{\textstyle f\left(\frac{x+y}{2},\eta-A\left(\frac{x+y}{2}\right)\right)}\;\!u(y).
$$

The operator norm $\|\cdot\|$ on $\BB(\H)$ being relevant in Quantum Mechanics, we
pull it back by setting
$$
\|\cdot\|^B_\hb : \S(\Xi)\rightarrow\R_+ \quad {\rm with}\quad
\|f \|^B_\hb:=\| \Op(f) \|.
$$
By gauge covariance, it is clear that $\|\cdot\|^B_\hb$ only depends on the magnetic field
$B$ and not on the vector potential $A$. We denote by $\ABh$ the completion of
$\S(\Xi)$ under $\|\cdot\|^B_\hb$. It is a $C^*$-algebra that can be identified with a vector subspace of
$\S^*(\Xi)$ and $\Op:\ABh \rightarrow \BB(\H)$ is a faithful $^*$-representation, with
$\Op\big[\ABh\big] = \K(\H)$, the $C^*$-algebra of compact operators in $\H$.

Many other useful $C^*$-algebras can be defined in this manner. An important one is $\BBh$, defined such
that $\Op:\BBh\rightarrow \BB(\H)$ is an isomorphism. The "magnetic version " of the
Calderon-Vaillancourt theorem, proved in \cite{IMP}, says that if $B_{jk}\in BC^\infty (\X)$ for $j,k = 1,\dots,N$,
then the Fr\'echet space $BC^\infty (\Xi)$ of smooth functions on $\Xi$ having
bounded derivatives of any order is continuously embedded in $\BBh$. On the other hand, $C^*$-algebras modelling
magnetic pseudodifferential operators "with anisotropic coefficients" are available by twisted crossed product techniques.
In this setting, our $\ABh$ is isomorphic to $C_0 (\X){\rtimes}^{{\omega}^B}_\tau \X$;
we refer to \cite{LMR,MPR1,MPR2} for details and applications.

To summarize, for any $\hb \in (0,1]$ we have defined a $C^*$-algebra $\ABh$ embedded into
$\S^*(\Xi)$ and isomorphic by $\Op$ to the $C^*$-algebra of compact operators on $\H$.
The product is essentially given by \eqref{composition}.

To justify \eqref{op} let us define a family $(\e_X)_{X\in\Xi}\subset C^\infty_{\rm{pol,u}}(\Xi)$ by $\e_X(Z):=e^{-i\sigma(X,Z)}$.
For suitable functions $f:\Xi\rightarrow\C$ and with suitable interpretation, one has
\begin{equation*}
f(Y)=(2\pi)^{-N}\int_\Xi \d X\;\![\mathfrak F_\Xi f] (X)\;\!e^{-i\sigma(X,Y)}
=(2\pi)^{-N}\int_\Xi \d X\;\![\mathfrak F_\Xi f](X)\;\! \e_X(Y)\ ,
\end{equation*}
where $\mathfrak F_\Xi f$ is \emph{the symplectic Fourier transform of $f$}. So a good quantization should have the property
\begin{equation}\label{just}
\Op(f)=(2\pi)^{-N}\int_\Xi \d X\;\![\mathfrak F_\Xi f](X)\;\! \Op(\e_X).
\end{equation}
Thus, the problem is to justify a choice for the operators $W^A_\hb(X):=\Op(\e_X)$ acting in $\H$.

In the presence of a magnetic field $B=\d A$, a basic family of self-adjoint operators is
$$
\big(Q,\Pi^A_\hb\big)\equiv \big(Q_1,\dots,Q_N, \Pi^A_{\hb,1},\dots,\Pi^A_{\hb,N}\big)\ ,
$$
where $Q_j$ is the operator of multiplication by the coordinate function $x_j$ and
$\Pi^A_{\hb,j}:=-i\hb\partial_j-A_j$ is the $j$-th component of the magnetic momentum. They satisfy the commutation relations
\begin{equation*}
i[Q_j,Q_k]=0,\quad i[\Pi^A_{\hb,j},Q_k]=\hb\;\!\delta_{j,k},\quad i[\Pi^A_{\hb,j},\Pi^A_{\hb,k}]=-\hb\;\! B_{jk}\ .
\end{equation*}
Then, by admitting that the quantization of the function $X\mapsto\e_Y(X)$ should be the unitary operator
$$
\exp\big[-i\sigma\big((y,\eta),(Q,\Pi^A_{\hb})\big)\big]=\exp\big[-i(Q\cdot\eta-y\cdot\Pi^A_{\hb})\big],
$$
one obtains by an explicit computation (relying for instance on Trotter's formula) {\it the magnetic Weyl system}
\begin{equation}\label{tro}
W^A_\hb(y,\eta)=e^{-i\,(Q+\frac{\hb}{2}y)\cdot\eta}\;\!e^{-\frac{i}{\hb}\Gamma^A[Q,Q+{\hb}y]}\;\!e^{i\,y\cdot\hb\,D},
\end{equation}
with $D_j:=-i\partial_j$. When applied to $u \in \H$ and for $x \in \X$, it explicitly gives
$$
[W^A_\hb (Y)u](x)=e^{-i(x+\frac{\hb}{2}y)\cdot\eta}\;\!e^{-\frac{i}{\hb}\Gamma^A[x,x+\hb
y]}\;\!u(x+\hb\;\!y) \ .
$$
Note that for $A=0$ one recognizes the usual Weyl system $W_\hb(y,\eta)$.
Finally, by plugging this into \eqref{just} leads to \eqref{op}, which may be considered as one of the possible justifications of the
formalism.

\section{Magnetic coherent states}\label{oita}

Let us fix a unit vector $v\in\H := L^2 (\X)$, and for any $\hb \in I:=(0,1]$
we define the unit vector $v_\hb\in\H$ by $v_\hb(x):=\hb^{-N/4}v\big(\frac{x}{\sqrt{\hb}}\big)$.
Using the non-magnetic Weyl system, for every $Y\in\Xi$ we set $v_\hb(Y):=W_\hb(-Y/\hb)\;\!v_{\hb}$.
For any choice of a continuous vector potential $A$ generating the magnetic field $B$, we then define
$$
v^A_{\hb}(Z):=e^{\frac{i}{\hb}\;\!\Gamma^A[z,Q]}\;\!v_\hb(Z)=e^{\frac{i}{\hb}\;\!\Gamma^A[z,Q]}\;\!
W_\hb(-Z/\hb)\;\!v_{\hb}.
$$
Explicitly, this means
\begin{equation}\label{mur}
\left[v^A_{\hb}(Z)\right](x)=e^{\frac{i}{\hb}(x-\frac{z}{2}) \cdot\zeta}
\;\!e^{\frac{i}{\hb}\Gamma^A[z,x]}\;\!v_\hb (x-z).
\end{equation}

The pure state space of $\K(\H)$ can be identified with the projective space $\P(\H)$;
considering the isomorphism
$\Op:\ABh\rightarrow \K(\H)$, it is natural to introduce the following families
of pure states on the two $C^*$-algebras:

\begin{definition}\label{enfin}
For any $Z\in\Xi$ we define $\v^A_{\hb}(Z):\mathcal K(\H)\rightarrow\C$ by
$$
\left[\v^A_{\hb}(Z)\right](S):=
\Tr\big(\left\vert v^A_{\hb}(Z)\right> \left< v^A_{\hb}(Z) \right\vert S \big)\equiv
\left< v^A_{\hb}(Z),S\;\!v^A_{\hb}(Z)\right>,
$$
for any $S\in\K(\H)$, and $\v^B_{\hb}(Z):\ABh\rightarrow\C$ by
$$
\left[\v^B_{\hb}(Z)\right](f):=\left[\v^A_{\hb}(Z) \right]\big(\Op(f)\big)=
\big< v^A_{\hb}(Z),\Op(f)\;\!v^A_{\hb}(Z)\big>
$$
for any $f\in\ABh$.
\end{definition}

The intrinsic notation $\mathfrak v^B_{\hb}(Z)$ is justified by a straightforward computation
leading for $Z = (z,\zeta)$ to
\begin{eqnarray}\label{insfarsit}
&&\left[\v^B_{\hb}(Z)\right](f) \\
\nonumber &=&(2\pi\hb)^{-N}\int_\X\int_\X\int_{\X^*}\d x\,\d y\,\d\eta\,
e^{\frac{i}{\hb}(x-y)\cdot(\eta-\zeta)}{\textstyle f\left(\frac{x+y}{2},\eta\right)}
\;e^{-\frac{i}{\hb}\Gamma^B\langle z,x,y\rangle}\;\!\overline{v_\hb(x-z)}\;\!v_\hb(y-z)\\
\nonumber &=&(2\pi\hb)^{-N}\int_\X\int_\X\int_{\X^*}\d x\,\d y\,\d\eta\,
e^{\frac{i}{\hb}(x-y)\cdot\eta}{\textstyle f\left(z+\frac{x+y}{2},\zeta+\eta\right)}
\;e^{-\frac{i}{\hb}\Gamma^B\langle z,x+z,y+z\rangle}\;\!\overline{v_\hb(x)}\;\!v_\hb(y) \\
\nonumber &=&(2\pi)^{-N}\int_\X\int_\X\int_{\X^*}\d x\,\d y\,\d\eta\,
e^{i(x-y)\cdot\eta}{\textstyle f\left(z+\frac{\sqrt\hb}{2}(x+y),\zeta+\sqrt\hb\eta\right)}
\;e^{-\frac{i}{\hb}\Gamma^B\langle z,z+\sqrt\hb x,z+\sqrt\hb y\rangle}\;\!\overline{v(x)}\;\!v(y)\ .
\end{eqnarray}

\begin{definition}\label{bunbun}
The family $\{v^A_{\hb}(Z)\,|\,Z\in\Xi,\,\hb\in I\}$ given by \eqref{mur} is called
\emph{the family of magnetic coherent vectors associated with the pair} $(A,v)$.
The elements of the families $\{\v^A_{\hb}(Z)\,|\,Z\in\Xi,\,\hb\in I\}$ and
$\{\v^B_{\hb}(Z)\,|\,Z\in\Xi,\,\hb\in I\}$ will be called \emph{the coherent states}.
\end{definition}

\begin{remark}\label{emark}
{\rm Our strategy for defining coherent states is quite remote of the standard one, consisting in propagating a given state along the orbit of a (projective) representation. Using for this purpose the magnetic Weyl system \eqref{tro}, which is a very general type of projective representation, would have given an $A$-dependent family of states on the intrinsic $C^*$-algebras $\mathfrak A^B_\hb$, which cannot be admitted. We tried other choices, which lead to $B$-depending states but which have weaker properties concerning the limit $\hb\rightarrow 0$. It is not clear to us how to put our approach in the perspective of covariant quantization of phase space \cite{AAGM,La4}.}
\end{remark}

\begin{remark}\label{rima}
{\rm Note however that for the standard Gaussian $v(x)=\pi^{-N/4}e^{-x^2/2}$ and
for $A=0$ one gets the usual coherent states of Quantum Mechanics
(see for example \cite{AAGM}). We insist of the
fact that the state corresponding to $Z=0$ is built upon an arbitrary unit vector of $\H$.
The standard Gaussian choice (generating a holomorphic setting in the absence of a
magnetic field) has no relevance at this stage. For part of our results, however,
some smoothness or decay properties will be needed.}
\end{remark}

The first result, basic to any theory involving coherent states, says that
\begin{equation*}
\int_{\Xi}\frac{\d Y}{(2\pi\hb)^N}\; |v^A_{\hb}(Y)\rangle\langle v^A_{\hb}(Y)|=1.
\end{equation*}
For convenience, we shall sometimes use the shorter notation
$\gamma^A_\hb(x,y)$ for $e^{-\frac{i}{\hb}\Gamma^A[x,y]}$.

\begin{proposition}\label{ax1}
Assume that the magnetic field $B$ is continuous and let $v$ be a unit vector in $\H$.
For any $\hb\in I$ and $u\in\H$ with $\|u\|=1$, one has
\begin{equation}\label{par}
\int_{\Xi}\frac{dY}{(2\pi\hb)^N}\;\big|\langle v^A_{\hb}(Y),\,u\rangle\big|^2=1.
\end{equation}
\end{proposition}

\begin{proof}
One has to show that
$$
\int_\Xi\frac{\d Y}{(2\pi\hb)^N}\big|\langle v^A_{\hb}(Y),\,u\rangle\big|^2=
\| v^A_{\hb}\|^2\;\!\|u\|^2,
$$
which follows if the mapping
\begin{equation}\label{mar}
\H\otimes\H\ni v\otimes u\mapsto \big\langle e^{\frac{i}{\hb}\;\!\Gamma^A[\cdot,Q]}\;\!
W_\hb(-\cdot/\hb)
\;\!v,u\big\rangle\in L^2\Big(\Xi;\frac{\d Y}{(2\pi\hb)^N}\Big)
\end{equation}
is proved to be isometric.

By using \eqref{mur} and after a simple change of variables, one gets:
\begin{eqnarray*}
\left\langle \bar \gamma^A_\hb(y,Q)\;\!W_\hb (-Y/\hb)\;\!v,u \right\rangle
&=& \int_\X \d x\,e^{-\frac{i}{\hb} x\cdot\eta}\;\!{\textstyle \gamma^A_\hb
(y,x+y/2)\;\!u(x+y/2)\;\!\overline{v}(x-y/2)} \\
&=& (2\pi)^{N/2}\,[(1\otimes\F)\circ C]\,[\beta^A_\hb \cdot(u\otimes\overline{v})]{\textstyle \left(y,\,\frac{\eta}{\hb}\right)},
\end{eqnarray*}
where $1\otimes\mathcal F$ is a unitary partial Fourier transformation, $C$ is the operator "change of variables"
$$
{\textstyle
(C F)(y,x):=F\left(x+\frac{y}{2},\,x-\frac{y}{2}\right)},
$$
which is also unitary in $L^2(\X\times\X)\simeq\mathcal H\otimes\mathcal H$ and
$\beta^A_\hb(x,y)=\gamma^A_\hb(x-y,x)$.
From this, and since $|\beta^A_\hb(x,y)|=1$ for all $x,y\in\X$, the isometry of
\eqref{mar} follows immediately.
\end{proof}

\begin{proposition}\label{probtranz}
Assume that the magnetic field $B$ is continuous and let $v$ be a unit vector in $\H$.
For any $Y,\,Z\in \Xi$, one has
$$
\lim_{\hb\rightarrow 0}\big|\langle v^A_{\hb} (Z),\,v^A_{\hb} (Y)\rangle\big|^2=\delta_{ZY}.
$$
\end{proposition}

\begin{proof}
Since the case $Z=Y$ is trivial, we can assume that $Z\neq Y$.
A short computation shows that the expression
$\langle v^A_{\hb} (Z),\,v^A_{\hb} (Y)\rangle$ is equal to
$$
e^{\frac{i}{2\hb}(z\cdot\zeta-y\cdot\eta)}\;\!\gamma^A_\hb(z,y)
\int_\X \d x\;\!e^{\frac{i}{\hb}x\cdot(\eta-\zeta)}\;\!
e^{-\frac{i}{\hb}\Gamma^B\langle x,y,z\rangle}\;\!
\hb^{-N/2}\;\!{\textstyle \overline{v\left(\frac{x-z}{\sqrt{\hb}}\right)}\;\!
 v\left(\frac{x-y}{\sqrt{\hb}}\right)}.
$$
After the change of variables $\frac{x-z}{\sqrt{\hb}}\,=\,t$, one gets
\begin{equation}\label{por}
\big|\langle v^A_{\hb} (Z), v^A_{\hb} (Y)\rangle\big|
= \left| \int_\X \d t\;\!
e^{it\cdot\frac{\eta-\zeta}{\sqrt{\hb}}} \;\!e^{-\frac{i}{\hb}
\Gamma^B\langle z+\sqrt{\hb}t,y,z\rangle}\;\!
\overline{v(t)}\;\!{\textstyle v\left(t+\frac{z-y}{\sqrt{\hb}}\right)}\right|.
\end{equation}

Now, if $z\neq y$, the r.h.s.~of \eqref{por} is dominated by $\int_\X \d t \;\!|v(t)|\;\! |v\left(t+\frac{z-y}{\sqrt{\hb}}\right)|$. It is easily shown that this one converges to $0$ as $\hb\rightarrow 0$ by a simple approximation argument using functions with compact support.
On the other hand, if $z=y$ but $\eta\neq\zeta$, then the r.h.s.~of \eqref{por} is equal to
$$
\left|\int_\X \d t\;\! e^{it\cdot\frac{\eta-\zeta}{\sqrt{\hb}}}|v(t)|^2\right|
=(2\pi)^{N/2}\left|(\F^* |v|^2){\textstyle \left(
\frac{\eta-\zeta}{\sqrt{\hb}}\right)}\right|
$$
which converges to $0$ as $\hb\rightarrow 0$ by the Riemann-Lebesgue Lemma.
\end{proof}

The next property of the magnetic coherent states is an important one, having consequences
on the behavior of the magnetic Berezin quantization for $\hb\rightarrow 0$.
Unfortunately, some extra assumptions on the magnetic field $B$ and on the state $v$ will be needed.
On the other hand, notice that the result is valid
for any bounded continuous function $g$; no compact support assumption is required (see \cite[Sec.~II.1.3]{La3}).

\begin{proposition}\label{debaza}
Let $B_{jk}\in BC^\infty(\X)$ for $j,k\in \{1,\dots,N\}$ and assume that $v\in\S(\X)$.
For any $g:\Xi\rightarrow\C$ bounded continuous function and any $Z\in \Xi$ one has
$$
\lim_{\hb\rightarrow 0}\int_\Xi\frac{\d Y}{(2\pi\hb)^N}\;\!\big|\langle v^A_{\hb}(Z),
v^A_{\hb}(Y)\rangle\big|^2 \,g(Y)\,=\,g(Z)\ .
$$
\end{proposition}

\begin{proof}
By taking \eqref{par} into account one has to show that
\begin{equation}\label{nor}
\int_\Xi \frac{\d Y}{(2\pi\hb)^N}\big|\langle v^A_{\hb}(Z),v^A_{\hb}(Y)\rangle|^2 [g(Y)-g(Z)]
\end{equation}
converges to $0$ for $\hb\rightarrow 0$. By using the expression already obtained in \eqref{por} together with the change of variables $x = (z-y)/\sqrt{\hb}$ and $\xi = (\eta-\zeta)/\sqrt{\hb}$, one gets that \eqref{nor} is equal to
\begin{equation}\label{tur}
(2\pi)^{-N}\int_\X \int_{\X^*} \d x \;\!\d \xi\;\! \big[g\big(z-\sqrt{\hb}\;\!x,\,\zeta+\sqrt{\hb}\;\!\xi\big)-g(z,\zeta)\big] \big| F_{\hb,z}(x,\xi)\big|^2,
\end{equation}
with
\begin{equation*}
F_{\hb,z} (x,\xi):=\int_\X \d t\;\!
e^{it\cdot \xi}\;\!
e^{-\frac{i}{\hb}
\Gamma^B\langle z, z+\sqrt{\hb}\;\!t,z-\sqrt{\hb}\;\!x\rangle}
\;\!\overline{v(t)}\;\!v(t+x)\ .
\end{equation*}

We are now going to apply the Dominated Convergence Theorem to show that \eqref{tur} converges to $0$ when $\hb\rightarrow 0$. First, since $g$ is continuous one has for any fixed $(x,\xi)\in \Xi$ that
$$
\lim_{\hb \to 0 } \big[g\big(z-\sqrt{\hb}\;\!x,\zeta+\sqrt{\hb}\;\!\xi\big)-g(z,\zeta)\big] =0\ .
$$
In addition, for any $\hb\in I$ and any $(x,\xi)\in \Xi$ one has $\big|F_{\hb,z}(x,\xi)\big|^2\leq\| v\|^4$,
so the integrand of \eqref{tur} converges to $0$.

Since $g$ is also bounded, it will be enough to find functions $a\in L^1(\X;\R_{+})$, $b\in L^1(\X^*;\R_{+})$
 such that for all $\hb,x$ and $\xi$
\begin{equation}\label{sin}
|F_{\hb,z}(x,\xi)|\leq a(x)\ \ {\rm and}\ \ |F_{\hb,z}(x,\xi)|\leq b(\xi).
\end{equation}

The first estimation is simple. For arbitrary positive numbers $n$, $m$ one has
$$
|F_{\hb,z}(x,\xi)|\leq\int_{\X}\d t\;\!|v(t)|\;\!| v(t+x)|\leq C\int_{\X}\d t\;\!\langle t\rangle^{-m}\;\!\langle t+x\rangle^{-n}\leq
C'\Big[\int_{\X}\d t\;\! \langle t \rangle^{-m} \langle t\rangle^{n}\Big]\langle x\rangle^{-n}\ ,
$$
and it is then enough to choose $n>N$ and $m>N+n$.

The second inequality in \eqref{sin} is more involved. Relying on a basic property of Fourier transformation, it consists essentially in showing that the function
$$
w_\hb(t;z,x):= e^{-\frac{i}{\hb}
\Gamma^B\langle z, z+\sqrt{\hb}\;\!t,z-\sqrt{\hb}\;\!x\rangle}
\;\!\overline{v(t)}\;\!v(t+x)
$$
belongs to the Schwartz's space with respect to $t$, uniformly with respect to the parameters $z,x,\hb$.
First of all, since $v\in \S(\X)$, for any $\alpha,\beta\in\N^N$ and any $k,\,l\in\N$ one has
$$
\big|\overline{(\partial^{\alpha}v)(t)}\;\!(\partial^{\beta}v)(t+x)\big|\leq C^{\alpha,\beta}_{k,l}\;\!\langle t\rangle^{-k}\langle x\rangle^{-l}.
$$

Thus it will be enough to show that the map $e^{-\frac{i}{\hb}
\Gamma^B\langle z, z+\sqrt{\hb}\;\!\cdot,z-\sqrt{\hb}\;\!x\rangle}:\X \to \mathbb C$ is in $C^{\infty}_{\rm pol}(\X)$, and that the polynomial bounds on each of its derivatives are uniform in $\hb\in I$ and $z\in\X$, with growth at most polynomial in $x\in\X$. For this, we need to recall the natural parametrizations of the flux of the magnetic field $B$:
\begin{equation*}
\Gamma^B\langle a,b,c\rangle
= \sum_{j,k}(b_j-a_j)\;\!(c_k-b_k)
\int_0^1 \d \mu\int_0^1\d \nu\;\!\mu\;\!B_{jk}
\big(a+\mu(b-a)+\mu\nu (c-b)\big).
\end{equation*}
By using this parametrization, one obtains
$$
e^{-\frac{i}{\hb}
\Gamma^B\langle z, z+\sqrt{\hb}\;\!t,z-\sqrt{\hb}\;\!x\rangle}
=\exp\Big\{i\sum_{j,k}
t_j\;\! (x_k+t_k) \int_0^1 \d \mu\int_0^1\d \nu\;\!\mu\;\!B_{jk}
\big(z+\mu \sqrt{\hb}\;\!t - \mu\nu \sqrt{\hb}\;\!(x+t) \big)\Big\}.
$$
The needed estimates follow then quite straightforwardly from this representation and from the assumption that all the derivatives of the magnetic field are bounded.
\end{proof}

We now take into considerations the maps
$v^A_{\hb}:\left(\Xi,\sigma^B\right)\rightarrow(\H,\Sigma_\hb)$ and
$\v^A_{\hb}:\left(\Xi,\sigma^B\right)\rightarrow\left(\P(\H) ,\Sigma'_\hb\right)$
between symplectic manifolds.
We refer to \cite[Sec.~I.2.5]{La3} for a detailed presentation of the symplectic structures of $\H$ and of the projective space $\P(\H)$ and simply recall some key elements.
Note that our convention differs from that reference by a minus sign.

On the Hilbert space $\H$, the (constant) symplectic form is defined at the point $w \in \H$ by
$$
\Sigma_{\hb,w}(u,v):=-2\hb\im \langle u,v\rangle
$$
for each $u,v\in \H$ and $\hb \in I$.
For the space $\P(\H)$, recall first that each of its elements can be identified with the one-dimensional orthogonal projections $\v=|v\rangle \langle v|$ defined by some unit vector $v\in\H$, with the known phase ambiguity. Then, the Fubini-Study symplectic form $\Sigma'_\hb$ is explicitly given at the point $\v \in \P(\H)$ by
\begin{equation*}
\Sigma'_{\hb,\v}(iSv,iTv)=i\hb\,\v([S,T]) =i\hb\langle v,[S,T]v\rangle
\end{equation*}
for any self-adjoint element $S,T$ of $\B(\H)$.
This relies among others on identifying the tangent space $T_\v\P(\H)$ to the real vector space $\{iSv\mid S=S^*\in\B(\H)\}$.

We would like now to show that the pull-back by $\v^A_{\hb}$ of the form $\Sigma'_\hb$ converges to $\sigma^B$ when
$\hb\rightarrow 0$. But $\Sigma_\hb$ is already the pull-back of $\Sigma'_\hb$ by the canonical map
$:\H\rightarrow\P(\H)$, so we only need to show the next result:

\begin{proposition}\label{pulbec}
Let $B_{jk}\in BC^\infty(\X)$ for $j,k\in \{1,\dots,N\}$ and assume that $v\in\S(\X)$.
The pull-back by $v^A_{\hb}$ of $\Sigma_\hb$ converges to $\sigma^B$ in the limit
$\hb\rightarrow 0$.
\end{proposition}

\begin{proof}
One has to calculate, for any $X,Y,Z\in\Xi$, the expression
\begin{equation*}
\Sigma_{\hb,v^A_{\hb}(X)}\left(T[v^A_{\hb}(X)](Y),T[v^A_{\hb}(X)](Z)\right)
=-2\hb\im\big\langle T[v^A_{\hb}(X)](Y),T[v^A_{\hb}(X)](Z)\big\rangle
\end{equation*}
where $T$ denotes the tangent map (total derivative). For that purpose, let us recall that
$$
v^A_{\hb}(X)=e^{\frac{i}{\hb}\Gamma^A[x,Q]}\;\!W_\hb(-X/\hb)\;\!v_{\hb}
= e^{\frac{i}{\hb}(Q-\frac{x}{2})\cdot \xi }\;\!e^{\frac{i}{\hb}\Gamma^A[x,Q]}\;\!e^{-ix\cdot D}\;\!v_\hb\ .
$$
It then follows that
\begin{equation*}
T[v^A_{\hb}(X)](Y) = i\;\!e^{\frac{i}{\hb}(Q-\frac{x}{2})\cdot \xi }\;\!e^{\frac{i}{\hb}\Gamma^A[x,Q]}\;\!e^{-ix\cdot D}\;\!M_\hb(Y;Q,D)\;\! v_\hb
\end{equation*}
with (obvious formal notations)
\begin{equation*}
M_\hb(Y;Q,D):={\textstyle \frac{1}{\hb}\left[
-\frac{y}{2}\cdot \xi + \left(Q+\frac{x}{2}\right)\cdot \eta
+\frac{\partial \Gamma^A[x,q]}{\partial x}\big|_{q=Q+x}(Y)
\right]-y\cdot D}
\end{equation*}
and with
\begin{equation*}
\frac{\partial \Gamma^A[x,q]}{\partial x}\big|_{q=Q+x}(Y) =
-y\cdot \int_0^1\d s\;\!A\big(x+s\;\!Q\big)
+\sum_{j,k}Q_j\;\!y_k\int_0^1\d s\;\!(1-s)\;\!\partial_k A_j\big(x+s\;\!Q\big)\ .
\end{equation*}
A similar expression holds for $T\big[v^A_{\hb}(X)\big](Z)$.
Thus, one has to calculate
\begin{eqnarray}\label{yes}
\nonumber &&-2\hb\im\big\langle T\big[v^A_{\hb}(X)\big](Y),T\big[v^A_{\hb}(X)\big](Z) \big\rangle \\
\nonumber &=&-2\hb\im\big\langle
M_\hb(Y;Q,D)\;\! v_\hb,M_\hb(Z;Q,D)\;\! v_\hb
\big\rangle \\
\nonumber &=&-2\hb\im\left\langle
{\textstyle M_\hb\left(Y;\sqrt{\hb}\;\!Q,\frac{1}{\sqrt{\hb}}D\right)\;\! v,
M_\hb\left(Z;\sqrt{\hb}\;\!Q,\frac{1}{\sqrt{\hb}}D\right)\;\! v
}\right\rangle \\
&=&i\hb\left\langle
{\textstyle v, \left[M_\hb\left(Y;\sqrt{\hb}\;\!Q,\frac{1}{\sqrt{\hb}}D\right),
M_\hb\left(Z;\sqrt{\hb}\;\!Q,\frac{1}{\sqrt{\hb}}D\right)\right]\;\! v
}\right\rangle \ .
\end{eqnarray}
A rather lengthy calculation then gives
\begin{eqnarray*}
&&{\textstyle
\left[M_\hb\left(Y;\sqrt{\hb}\;\!Q,\frac{1}{\sqrt{\hb}}D\right),
M_\hb\left(Z;\sqrt{\hb}\;\!Q,\frac{1}{\sqrt{\hb}}D\right)\right] }\\
&=& -\frac{i}{\hb}\;\!\sigma(Y,Z) -\frac{i}{\hb} \sum_{j,k}y_j\;\!z_k\;\! \int_0^1\d s\;\!s\;\!
B_{jk}(x+s\;\!\sqrt{\hb}\;\!Q) \\
&& -\frac{i}{\hb} \sum_{j,k}y_j\;\!z_k\;\! \int_0^1\d s\;\!(1-s)\;\!
B_{jk}(x+s\;\!\sqrt{\hb}\;\!Q) \\
&&+\frac{i}{\hb}\;\!\sqrt{\hb}\sum_{j,k,\ell} (y_\ell \;\!z_k-y_k\;\!z_\ell)\;\!Q_j
\int_0^1 \d s \;\!s\;\!(1-s)\;\! \partial^2_{\ell k}A_j(x+s\;\!\sqrt{\hb}\;\!Q)\ .
\end{eqnarray*}
Finally, by inserting these expression into \eqref{yes} and by an application of the Dominated Convergence Theorem, one obtains the statement of the Proposition.
\end{proof}

The following result can be interpreted as the convergence of the quantum pure state $\v^B_\hb(Z)$
to a corresponding classical pure states in the semiclassical limit.
For that purpose, we shall consider functions $g:\overline I\times\Xi$ and write $g_\hb(X)$ for $g(\hb,X)$.
We assume that $\hb\mapsto g_\hb(X)$ is continuous for any $X\in\Xi$ and that $g_\hb\in\S(\Xi)$ for all $\hb$.

\begin{proposition}\label{ultima}
Let $B_{jk}\in BC^\infty(\X)$ for $j,k\in \{1,\dots,N\}$ and assume that
$v \in \S(\X)$. Then for any $g$ as above and any $Z\in\Xi$ one has
\begin{equation}\label{ultimita}
\underset{\hb\rightarrow 0}{\lim}\left[\v^B_{\hb} (Z)\right](g_\hb)=\delta_{Z}(g_0)=g_0(Z)\ .
\end{equation}
\end{proposition}

\begin{proof}
Let $\hb \in I$.
Starting from the last expression obtained in \eqref{insfarsit} and performing the change of variables $y' = x-y$, one finds that $\left[\v^B_{\hb}(Z)\right](g_\hb)$ is equal to
\begin{equation}\label{futurD}
\frac{1}{(2\pi)^{N}}\int_\X\int_\X\int_{\X^*}\d x\,\d y\,\d\eta\,e^{iy\cdot\eta}\;\!
g_\hb\big(z+\sqrt{\hb}\;\!(x-y/2),\zeta+\sqrt{\hb}\;\!\eta\big)
\;\!\varphi_\hb(z;x,y)\;\!v(x-y)\;\!\overline{v(x)}
\end{equation}
with $\varphi_\hb(z;x,y):=
e^{-\frac{i}{\hb}\Gamma^B\langle z,z+\sqrt{\hb}\;\!x,z+\sqrt{\hb}\;\!(x-y)\rangle}$.
The expression \eqref{futurD} can be rewritten as
\begin{equation}\label{futurD2}
\frac{1}{(2\pi)^{N/2}}\int_\X\d x \int_\X \frac{\d y}{\hb^{N/2}}
\; G_\hb\Big(x,y,\frac{y}{\sqrt{\hb}}\Big) \;\!\varphi_\hb(z;x,y)
\;\!v(x-y)\;\!\overline{v(x)},
\end{equation}
where
\begin{equation*}
G_\hb(x,y,\y):=  (1\otimes\F^*) [\Theta_1(z,\zeta)g_\hb]\big(\sqrt{\hb}\;\!(x-y/2),\y\big)
\end{equation*}
and $\Theta(z,\zeta)g_\hb=g_\hb(\cdot+y,\cdot+\zeta)$.
The map $\X \ni \y \mapsto G_\hb(x,y,\y) \in \C$ clearly belongs to $L^1(\X)$.

Now, in order to have a better understanding of the expression \eqref{futurD2}, let us observe that
\begin{equation*}
\varphi_\hb(z;x,y)= \exp\Big\{i\sum_{j,k}x_j\;\! y_k
 \int_0^1 \d \mu\int_0^1\d \nu\;\!\mu\;\!B_{jk}
\big(z+ \mu \sqrt{\hb}\;\!x  - \mu\nu \sqrt{\hb} \;\!y
\big)\Big\}\ .
\end{equation*}
Clearly, this functions has a limit as $\hb \to 0$.
More precisely, one has
\begin{equation*}
\lim_{\hb \to 0}\varphi_\hb(z;x,y) =
\exp\Big\{{\textstyle \frac{i}{2}}\sum_{j,k} y_k\;\! z_j\;\! B_{jk}(z)\Big\}
\end{equation*}
but one also has $\lim_{\hb \to 0}\varphi_\hb(z;x,\sqrt{\hb}\;\!y) =1$, both convergences being locally uniform.

By taking these information into account, one easily shows that
\begin{eqnarray}\label{futurD3}
&&\frac{1}{(2\pi)^{N/2}}\int_\X \frac{\d y}{\hb^{N/2}} \;G_\hb\Big(x,y,\frac{y}{\sqrt{\hb}}\Big)\;\!\varphi_\hb(z;x,y)\;\!v(x-y) \\
\nonumber &=&\frac{1}{(2\pi)^{N/2}}\int_\X \d y \;G_\hb\big(x,\sqrt{\hb}y,y\big) \;\!\varphi_\hb\Big(z;x,\sqrt{\hb}y\Big)
\;\!v\big(x-\sqrt{\hb}y\big)
\end{eqnarray}
converges as $\hb \to 0$ to
$$
(2\pi)^{-N/2}\int_\X \d y \;\!G_0(0,0,y)\;\!v(x)
= g_0\left(z,\zeta\right)v(x),
$$
locally uniformly in~$x$.

In order to conclude, one still has to show that for the expression \eqref{futurD2}
the limit $\hb \to 0$ and the integration with respect to $x$ can be exchanged.
However, this follows from the Dominated Convergence Theorem and the observation
that the expression \eqref{futurD3} is bounded in $x\in \X$ and in $\hb\in \overline{I}$.
\end{proof}

\section{Strict quantization}\label{oaie}

We start by recalling from \cite{MP2} the classical limit of the magnetic Weyl calculus. With the definitions and the notations introduced in Section \ref{vacoita}, and in the language of \cite{La3,Ri3,Ri4}, a
particular case of the results of \cite{MP2} states that

\begin{theorem}\label{sdq}
If $B_{jk}\in BC^\infty (\X)$ for $j,k = 1,\dots,N$, then the embeddings
$\big(\mathfrak Q^B_\hb:\S(\Xi;\R) \rightarrow \big[\ABh\big]_\R \big)_{\hb \in \bar I}$ form a strict
deformation quantization of the Poisson algebra $\big(\S(\Xi;\R),\ ,\{\cdot,\cdot\}^B\big)$.
\end{theorem}

To explain this, observe first that $\S(\Xi;\R)$ is really a Poisson subalgebra of $C^\infty(\Xi;\R)$
for the pointwise product and the Poisson bracket $\{\cdot,\cdot\}^B$ defined in \eqref{pison}.
The embedding $\mathfrak Q^B_\hb$ just interprets $f\in \S(\Xi;\R)$ as a self-adjoint element in the
$C^*$-algebra $\ABh$. The self-adjoint part $\big[\ABh\big]_\R$ is a
Jordan-Lie-Banach algebra under the $C^*$-norm and the operations
$$
\textstyle
f{{\circ}^B_\hb}g:=\frac{1}{2}\left({f{\sharp}^B_\hb g}+{g{\sharp}^B_\hb f}\right)\quad {\rm and}\quad
[f,g]^B_\hb:=\frac{1}{i\hb}\left({f{\sharp}^B_\hb g}-{g{\sharp}^B_\hb f}\right).
$$
Then the above theorem says that the following are verified for any $f,g\in \S(\Xi;\R)$:
\begin{enumerate}
\item $\bar I\ni \hb\mapsto \| f\|^B_\hb \in \R_{+}$ is continuous (Rieffel's axiom),
\item $\| \,f{\circ^B_\hb}g-fg\,\|^B_\hb \to 0$ as $\hb \to 0$ (von Neumann's axiom),
\item $\|\,[f,g]^B_\hb-\{f,g\}^B\,\|^{B}_{\hb} \to 0 $ as $\hb\to 0$ (Dirac's axiom).
\end{enumerate}

In our framework, a pure state quantization would be a family of smooth injections
$\{\v_{\hb}:\Xi\rightarrow \P(\H)\}_{\hb\in I}$ satisfying the three axioms stated in \cite[Def.~II.1.3.3]{La3}.
In fact, by setting
\begin{equation}\label{injection}
\v_{\hb}(X):=\v^A_{\hb}(X)\equiv |v_{\hb}^A(X)\rangle \langle v_{\hb}^A(X)|
\end{equation}
for any suitable unit vector $v \in \H$ and any $X \in \Xi$, the three axioms correspond our Propositions
\ref{ax1}, \ref{debaza} and \ref{pulbec}.
Actually, it seems to us that that the content of Proposition \ref{probtranz} is more intuitive than the one of Proposition \ref{debaza}, and could replace the latter statement.

Recalling the definition \eqref{tranzactia} of the quantum transition probabilities, let us still rewrite some of the results obtained so far in this language. In that framework, Proposition \ref{ax1} reads
\begin{equation}\label{parc}
\int_{\Xi}\frac{\d Y}{(2\pi\hb)^N}\;\!p^{\rm qu}\big(\v^A_{\hb}(Y) ,\mathfrak u\big)=1
\end{equation}
for any unit vector $v \in \H$, any $\hb \in I$ and any  $\mathfrak u\in\P(\H)$.
Proposition \ref{probtranz} is then equivalent to
\begin{equation}\label{mor}
\lim_{\hb\rightarrow 0}\;\! p^{\rm qu}\big(\v^A_{\hb} (Z),\v^A_{\hb} (Y)\big)= p^{\rm cl}(Z,Y)
\end{equation}
for any $Y,Z\in\Xi$.
Finally, under the stated conditions on $B_{jk}$ and $v$, Proposition \ref{debaza} reads
\begin{equation*}
\lim_{\hb\rightarrow 0}\int_\Xi\frac{\d Y}{(2\pi\hb)^N}\;\!p^{\rm qu}\big(\v^A_{\hb}(Z),
\v^A_{\hb}(Y)\big)\;\!g(Y)=g(Z)
\end{equation*}
for any $Z\in\Xi$ and $g\in BC(\Xi)$.

Obviously, in these relations one could replace $p^{\rm qu}$ and $\v^A_{\hb}$ with $p_\bullet^{\rm qu}$ and $\v^B_{\hb}$ simply by transporting the transition probabilities on the pure state space of the intrinsic algebra $\ABh$.
By collecting these results together with Proposition \ref{pulbec}  one has obtained :

\begin{theorem}\label{themainTh}
If $B_{jk}\in BC^\infty(\X)$ and $v\in\S(\X)$, the family $\{\v^A_{\hb}\}_{\hb\in I}$ forms a pure
state quantization of the Poisson space with transition probabilities $(\Xi, \sigma^B, p^{\rm cl})$.
\end{theorem}

\begin{remark}\label{ndsman}
{\rm We stress that the conditions \eqref{parc} and \eqref{mor} have been obtained for every
continuous magnetic field and for every unit vector $v$.
We also mention that the reference \cite{La3} imposes the axiom on the symplectic forms
as a limit for $\hb\rightarrow 0$, but says on page $114$ that in all the examples of this book
the equality holds without the limit. In might be interesting that we really need a limit.}
\end{remark}

We also set $\v^B_{\hb=0}(Z):=\delta_{Z}$. Let us now show that the family $\{\v^B_{\hb}(Z)\mid \hb\in \overline{I}, Z\in \Xi\}$
forms a continuous field of pure states associated with a continuous field of $C^*$-algebras, see
\cite[Sec.~II.1.2 \& II.1.3]{La3} for the abstract presentation.

We recall some results from \cite{MP2} and hinted in Theorem \ref{sdq}.
For $\hb \in I$ the $C^*$-algebra $\ABh$ is isomorphic (by a partial Fourier transform) to $\CC^B_\hb$,
the twisted crossed product algebras
$C_0(\X)\times^{\omega^B_\hb}_{\theta_\hb}\!\X$, where the group $2$-cocycle $\omega^B_\hb$ is described in \cite{MP2} and $[\theta_\hb(x)f](y)=f(y+\hb x)$ for any $x,y \in \X$ and $f \in C_0(\X)$.
Furthermore, let us consider the twisted action $\left(\Theta,\Omega^B\right)$ of $\X$ on $C_0(\overline{I}\times \X)$,
where $[\Theta(x)g](\hb,y):=g(\hb,y+\hb x)$ for all $g \in C_0(\overline{I}\times \X)$
and $[\Omega^B(x,y)](\hb,z):=\omega^B_\hb(z;x,y)$. The corresponding twisted crossed product algebra
$C_0(\overline{I}\times \X)\rtimes_\Theta^{\Omega^B}\!\!\X$ is simply denoted by $\CC^B$.
Now, it is proved in \cite[Sec.~VI]{MP2} that $\big(\CC^B,\{\CC^B_\hb,\varphi_\hb\}_{\hb \in \overline{I}}\big)$
is a continuous field of $C^*$-algebras, where $\varphi_\hb: \CC^B \to \CC^B_\hb$ is the surjective morphism
corresponding to the evaluation map $[\varphi_\hb(\Phi)](x)= \Phi(x,\hb)\in C_0(\X)$
for any $\Phi \in L^1\big(\X;C_0(\overline{I}\times \X)\big)$. By performing the partial Fourier transform,
one again obtains a continuous field of $C^*$-algebras $\big(\A^B,\{\ABh,\psi_\hb\}_{\hb \in \overline{I}}\big)$.
In this representation, the $C^*$-algebra $\A^B_0$ corresponding to $\hb = 0$ is simply equal to $C_0(\Xi)$.

\begin{proposition}\label{zizia}
Let $B_{jk}\in BC^\infty(\X)$ for $j,k\in \{1,\dots,N\}$ and assume that
$v \in \S(\X)$. Then the family $\{\v^B_{\hb}(Z)\mid \hb\in \overline{I}, Z\in \Xi\}$
forms a continuous field of pure states relative to the continuous field of $C^*$-algebras
$\big(\A^B,\{\ABh,\psi_\hb\}_{\hb \in \overline{I}}\big)$.
\end{proposition}

\begin{proof}
Since $\v^B_{\hb}(Z)$ is a pure state on $\ABh$ for each $Z \in \X$ and $\hb \in I$
and since $\delta_{Z}$ is a pure state on $\A^B_0 \equiv C_0(\Xi)$,
the proof simply consists in verifying that the two conditions stated in \cite[Def.~II.1.3.1]{La3}
are satisfied.

The first condition stipulates that for any $g \in \A^B$, the map $\hb \mapsto
[\v^B_{\hb}(Z)]\big(\psi_\hb(g)\big)$ belongs to $C(\overline{I})$.
If $g_\hb\in\S(\Xi),\,\forall \hb$, then the continuity in $\hb \in I$
follows from the explicit formula \eqref{insfarsit} and the continuity at $\hb = 0$
has been proved in Proposition \ref{ultima}.
The general case follows then by density and approximation.

The second condition requires that  $\cap_{Z\in \Xi}\;\!{\rm{ker}}[\pi^B_{\hb}(Z)]=\{0\}$ for any
$\hb\in \overline{I}$, where $\pi^B_{\hb}(Z)$ is the GNS representation associated with
$\v^B_{\hb}(Z)$. The GNS representation of a
pure state is irreducible, every irreducible representation of $\ABh\cong \K(\H)$
is unitarily equivalent to the identity representation, which is faithful.
So the mentioned condition holds for $\hb \in I$.
For $\hb =0$, this condition is also clearly satisfied.
\end{proof}

\begin{remark}\label{valdman}
{\rm According to the terminology of \cite{KNW}, the magnetic Weyl quantization is {\it positive}, meaning that any pure state $\delta_Z$ of the classical $C^*$-algebra $C_0(\Xi)$ can be deformed for each $\hb$ to a pure state $\v^B_{\hb}(Z)$ of the quantum $C^*$-algebra $\mathfrak A^B_\hb$ in the precise sense given by Proposition \ref{zizia} ({\it cf.}~also \eqref{ultimita}).
}
\end{remark}


\end{document}